\documentclass[12pt, reqno]{amsart}
\usepackage{times,epsfig}
\usepackage{amsfonts,amscd,amsmath,amssymb,amsthm}

\numberwithin{equation}{section}

\DeclareMathOperator{\nnab}{\! \mathop{\boldsymbol{\nabla}}\limits^{\longleftrightarrow}
\!\!}

\newcommand{\R}{\mathbb{R}}

\newcommand{\cK}{{\mathcal K}}

\newcommand{\diag}{\mathrm{diag}}

\newcommand{\bnab}{\boldsymbol{\nabla}\!}

\newcommand{\ii}{\mathrm{i}}
\newcommand{\e}{\mathrm{e}}

\newcommand{\rd}{\mathrm{d}}

\newcommand{\bA}{\mathbf{A}}
\newcommand{\ba}{\mathbf{a}}
\newcommand{\bb}{\mathbf{b}}
\newcommand{\bj}{\mathbf{j}}
\newcommand{\bc}{\mathbf{c}}
\newcommand{\bx}{\mathbf{x}}
\newcommand{\bv}{\mathbf{v}}
\renewcommand{\j}{\mathbf{j}}
\newcommand{\x}{\mathbf{x}}
\newcommand{\y}{\mathbf{y}}
\renewcommand{\u}{\mathbf{u}}
\renewcommand{\v}{\mathbf{v}}
\newcommand{\bw}{\mathbf{w}}

{\bf}{\it}
{\bf}{\it}
{\bf}{\it}
{\bf}{\it}
{\bf}{\it}

\newtheorem{theorem}{Theorem}[section]{\bf}{\it}
\newtheorem{proposition}[theorem]{Proposition}{\bf}{\it}
{\bf}{\it}
\newtheorem{example}[theorem]{Example}{\it}{\rm}
\newtheorem{lemma}[theorem]{Lemma}{\bf}{\it}
\newtheorem{remark}[theorem]{Remark}{\it}{\rm}
{\bf}{\it}
{\bf}{\it}
{\bf}{\it}

\title{Current Densities in Density Functional Theory}

\author[Elliott H. Lieb]{Elliott H. Lieb}
\address{Elliott H. Lieb\\ Princeton University, Departments of Physics and
Mathematics, \\ 
Jadwin Hall,
Princeton, NJ 08544 USA}
\email{lieb@princeton.edu}

\author[Robert Schrader]{Robert Schrader}
\address{Robert Schrader\\ Institut f\"{u}r Theoretische Physik\\
Freie Universit\"{a}t Berlin, Arnimallee 14\\ D-14195 Berlin, Germany}
\email{schrader@physik.fu-berlin.de}

\date {Aug. 20, 2013\ \ \ \ \  File: {E19density.tex} } 


\thanks{\copyright \, 2013 by the authors. This paper may be reproduced, in
its
entirety, for non-commercial purposes.}


\begin{document}

\begin{abstract}
It is well known that any given density $\rho(\bx)$ can be realized by a
determinantal wave function for $N$ particles. The question addressed here
is whether any given   density $\rho(\x) $   {\it and} current density
$\j(\bx)$ can be simultaneously realized by a (finite kinetic energy)
determinantal wave function.   In case the velocity field $\bv(\bx)
=\bj(\bx)/\rho(\bx)$ is curl free, we provide  a solution for all $N$, and
we provide an explicit upper bound for the energy.  If the velocity field
is not curl free, there is a finite energy solution for all $N\ge 4$,
but we do not provide an explicit energy bound in this case.  For $N=2$
we provide an example of a non curl free velocity field for which there
is a solution, and an example for which there is no solution. The case
$N=3$  with a non curl free velocity field is left open.

\end{abstract}

\maketitle


\section{Introduction}\label{sec:intro}
A question that arose in the early stages of density functional theory is
whether, given the one-body density $\rho(\x)$ of an $N$-body
system of fermions, there exists an $N$-body wave function (with finite 
kinetic energy)  whose reduced
one-body density equals the given one. More particularly, can this be 
accomplished with a determinantal wave function (under the 
obvious, necessary assumption, which will be made throughout, that $\bnab \sqrt{\rho}$ is 
square integrable). 

This article provides a proof of the existence of a fermionic
$N$-body {determinantal} state
with a given one-body density $\rho(\bx)$ and a given one-body current
density  $\bj(\bx)$ 
provided the velocity field $\bv(\bx) =\bj(\bx)/\rho(\bx)$ is curl
free. 
When $N\ge 4$, we prove the existence of 
solutions even if the velocity field is not curl free, as when there
are vortices, for example.
The proof is much more complicated in this case.
To avoid dwelling on unenlightening points of mathematical rigor
definitions of function spaces, smoothness, and other technical questions
are left to the reader. 
We do assume the obvious requirement that the support of $\bj$ is contained
in the support of $\rho$ and, for simplicity,  that $\bj$, $\rho$ and $\bv$ 
are differentiable. Actually we assume that $\rho$ and $\bv$ are given and thus 
we let the current be defined as $\bj=\rho\bv$.

In addition we provide a solution for an example with $N=2$
in which $\bv(\bx)$ is not curl free, which implies that 
`curl freeness'is not a necessary condition for finding a solution
when $N=2$. 
Again in the $N=2$ case, an example is provided for which {\it no} solution
exists.
This is contrary to a claim made without proof in \cite{GhoshDhara}, (see the 
discussion preceding relation (A.1) there) that there always is a solution. 
The same claim was made in \cite{Gross}, (see the sentence containing relation 
(54) there).
 Presently, it remains an open problem whether there always exist 
solutions when $N=3$.

To avoid possible confusion, we emphasize that we are 
discussing only the
existence of determinantal functions with the stated density and 
current; this state is not required to be the ground  state of any Hamiltonian.
We also note that we consider only spinless (i.e., spin-polarized)
particles here; spin can be included but it is an unnecessary complication here.

{\it Acknowledgements:} We thank E.U. Gross for making us aware of this problem
and for his encouragement. We are also grateful to Th. Br\"ocker, O.Lazarev
and 
V. Rutherfoord for helpful
discussions and suggestions, and we thank S. Kvaal for valuable comments on a first draft of
this paper. 
E.H.L. thanks the Simons Foundation for 
support through grant \#230207 and the U.S. National Science Foundation
for partial support through grant PHY-0965859.


\section{\bf Statement of the problem}\label{sec:problem}

{\it Notation:} In our units, $\hbar =1$ and the particle mass and charge 
are $m=1/2,\, -e=-1$.  Vectors are denoted by boldface. The {\it density} 
associated with a one-particle  function $\phi$ is 
given by $\rho(\x) = |\phi(\x)|^2$.
The {\it current density} is given by
$$
\bj_\phi(\x) = \frac{1}{2\ii} \left(\phi^* (\bx)  \bnab  \phi (\bx) - \phi(\bx)
\bnab \phi^* (\x) \right)\ :=  \frac{1}{2\ii} \phi^*(\x) \nnab\phi(\x) \ , 
$$
which also defines the symbol $\nnab$. This current is often called the
{\it paramagnetic current}. Clearly, 
$\int_{\R^3} \bnab \cdot \bj (\x) \rd\x =0$ by Green's theorem.
The actual physical current, in the 
presence of a magnetic
vector potential $\bA(\x)$, equals $\j(\x) + \bA(\x)\rho(\x).$ 
Since  $\rho(\x)$ and $\bA(\x)$ are regarded as given, the additional
$ \bA \rho $ term is thereby fixed and can be ignored for our considerations.

A fermionic $N$-body wave function $
\psi(\x_1,\x_2,\cdots, \x_N)$
is totally antisymmetric and normalized, i.e., 
\begin{equation}\label{pro1}
|| \psi||^2=\int_{
\R^{3N}}
| \psi(\x_1,\x_2,\cdots, \x_N)|^2 \rd\x_1 \rd\x_2\cdots \rd\x_N=1.
\end{equation}
As stated above, spin variables could, but will not be, included in our discussion.
The associated \emph{kinetic energy} is defined as 
\begin{equation}\label{pro2}
T(\psi)=\sum_{i=1}^{N}\int_{
\R^{3N}} | 
\bnab_{\x_i}
\psi(\x_1,\x_2,\cdots, \x_N)|^2
\rd\x_1\rd\x_2\cdots \rd\x_N.
\end{equation}
To this function $\psi$ we associate the \emph{one-body density}
\begin{equation}\label{pro3}
\rho_\psi(\x)=N\int_{
\R^{3(N-1)}} 
| \psi(\x,\x_2,\cdots, \x_N)|^2\rd\x_2\cdots \rd \x_N, 
\end{equation}
such that by \eqref{pro1}
\begin{equation}\label{pro4}
\int_{\R^3} \rho_\psi(\x)\rd\x =N.
\end{equation}
The associated current density is 
\begin{equation}
\bj(\x)=\frac{N}{2\ii}\int_{
\R^{3(N-1)}} 
\psi^{\star}(\x,\x_2,\cdots, \x_N)\nnab_{\x}\psi^{\star}
(\x,\x_2,\cdots, \x_N)
\rd\x_2\cdots \rd \x_N
\end{equation}
When $\phi_1(\x),\cdots,\phi_N(\x)$ are orthonormal functions on $\R^3$, the 
{\it $N$-body determinantal state}
\begin{equation}\label{pro6}
\psi(\x_1,\cdots,\x_N)=(N!)^{-1/2}\det\{\phi_k(\x_i)\}_{1\le i,k\le N}
\end{equation} 
is normalized and 
\begin{align}\label{pro7}
\rho_\psi(\x)&=\sum_{k=1}^{N}|\phi_k(\x)|^2 =  \sum_{k=1}^{N}\rho_k(\x)
\\\nonumber
{\bj}_\psi(\x)&=\frac{1}{2\ii}\sum_{k=1}^{N}\phi_k^{\star}(\x){\nnab}
\,\phi_k(\x)= \sum_{k=1}^{N} \bj_k (\x)\\\nonumber
T(\psi)&=\sum_{k=1}^{N}\int_{\R^3}|{\bnab}
\,\phi_k(\x)|^2 \rd\x \ .
\end{align}

We look for one-body functions that can be written as
\begin{equation}\label{so1}
\phi_k(\x)=\rho_k(\x)^{1/2}\,\e^{\ii\chi_k(\x)},
\end{equation} 
with a single-valued phase function $\chi$
and with the orthonormality property
\begin{align}\label{so2}
\langle \phi_l,\phi_k\rangle = \int_{\R^3}\rho_k(\x)^{1/2}\,\rho_l(\x)^{1/2}
\exp\{\ii(\chi_l(\x)-\chi_k(\x))\}\,\rd\x=\delta_{kl}.
\end{align} 
Our condition (\ref{so1}) should be noted. We are restricting ourselves to functions with a
well defined global phase. For example, the function $\psi(\x) = (x^1+ix^2) \e^{-|x|^2}$
is a real analytic function that solves the problem for a smooth $\bj$ and $\rho$, whose 
velocity field has a curl (a delta-function), yet it has no global phase function. The fact that
we can solve the problem for $N\geq 4$ with functions having a well defined phase is, therefore,
of some interest.

The second equation in \eqref{pro7} take the form
\begin{equation}\label{so301}
{\bj}(\x)=\sum_{k=1}^{N}\bj_k(\x)=\sum_{k=1}^{N}\rho_k(\x)\,\bnab\,\chi_k(\x).
\end{equation}
Our finite kinetic energy condition means that each component of the vector field
$\bnab\,\phi_k(\x)$ is square integrable.

Finally, we define two energies: The kinetic energy of a density $\rho(\x)$
\begin{equation}\label{pro701}
E(\rho)=\int_{\R^3} |\bnab \rho(\x)^{1/2}|^2
\,\rd \x
\end{equation}
and the { kinetic energy} of a current density $\bj$ 
\begin{align}\label{pro702}
E({\bj};\rho)=\int_{R^3}\frac{1}{\rho(\x)}\,|{\bj}(\x)|^2\,\rd \x \
=\int_{\R^3}\rho(\x)\,|{\bv}(\x)|^2\,\rd\x
\end{align}
with the {\it velocity field} 
\begin{equation}\label{so5}
{\bv}(\x)=\frac{1}{\rho(\x)}{\bj}(\x).
\end{equation}
One quickly checks that, for a determinantal function,
\begin{equation}\label{so3023}
T(\psi)=\sum_{k=1}^{N}(E(\rho_k)+E({\bj}_k;\rho_k)).
\end{equation}
This identity is the motivation for introducing the kinetic energy associated with
a density and with a current density.

Now we can formulate\\
{\bf The current-density problem:}  
 {\it  Given a density $\rho(\x)$ with $\int\rho(\x)\rd \x=N$ and 
a current density ${\bj}(\x)$, satisfying  
$\int_{\R^3} \bnab\cdot {\bj}(\x)\,\rd \x=0,$
is there an 
$N$-body determinantal state $\psi$, with functions as in (\ref{so1}), and with
$\rho_\psi(\x)=\rho(\x) \ {\rm and}\  {\bj}_\psi(\x)={\bj}(\x) $\, ?

Suppose, in addition, 
$E(\rho)+E(\bj;\rho)<\infty$.  
Can this state be chosen to satisfy $T(\psi)<\infty$? If
so, what bound can be placed on $T(\psi)$?}

For a physical motivation of this problem, see \cite{Gross}. For a previous 
discussion of this problem in the 1-dimensional case, see \cite{GhoshDhara,Gross}.

We recall a result  for the case in which the density 
$\rho(\x)$ alone is considered, that is no ${\bj(\x)}$ is prescribed, and hence
the second condition is merely $E(\rho)< \infty $. This was 
solved affirmatively, independently, and by the same method in \cite{Harriman} and \cite{Lieb}.
The solution happens, incidentally, to have the property that $\bj=0$.

The following bound appears in \cite{Lieb}:

Suppose $E(\rho)< \infty $. Then there is an $N$-body determinantal state
$\psi$ satisfying 
$\rho_\psi = \rho$ and
\begin{equation}
T(\psi)\le (4\pi)^2 N^2 E(\rho).
\end{equation}


\section{\bf Solution of the current-density problem for a curl free velocity field}
\label{sec:curlfree}

In this section we solve the problem for arbitrary $N\ge 1$, when the velocity field is curl free,
$\bnab\times \bv= 0$.  See Theorem \ref{4:prop:1}.  In section \ref{sec:N>3} we will drop 
this condition and will be able solve the problem when $N\ge 4$.
First, we recall the well known solution \cite{Harriman, Lieb} 
to the familiar problem of finding $\psi$ which solves  
$\rho_{\psi}=\rho$ for given $\rho$.
Write $\x=(x^1,x^2,x^3)$ and define for $-\infty<x^3<\infty$
\begin{equation}\label{so8}
f(x^3)=\frac{2\pi}{N}\int_{-\infty}^{\infty}\int_{-\infty}^{+\infty}
\int_{-\infty}^{x^3}\;\rho(s,t,u)\,\rd s\rd t\rd u
\end{equation}
which is monotone increasing from $0$ to $2\pi$. For given $N$ we introduce the set of $N$ numbers
\begin{equation}\label{so15}
\cK_N=\left\{-\frac{N-1}{2},\, -\frac{N-3}{2},\,  \cdots
,\frac{N-3}{2},\, \frac{N-1}{2}\right\}
\end{equation} 
satisfying
\begin{equation}\label{sumf}
\sum_{k\in\cK_N}k=0.
\end{equation} 
Set
\begin{equation}\label{so151}
C_N=\frac{16\pi^2}{N}\sum_{k\in\cK_N} k^2.
\end{equation} 
Define 
\begin{equation}\label{so9}
\phi_k(\x)=\left[\rho(\x)/N\right]^{1/2}\exp\left\{{\ii k
f(x^3)}\right\},\quad k\in\cK_N
\end{equation}
and $\psi(\x_1,\cdots, \x_N)=(N!)^{-1/2}\det \phi_k(\x_j)$.
Then $\rho_k(\x)=\rho(\x)/N$ and $\chi_k(\x)=kf(x^3)$.
These $N$ functions $\{\phi_k\}$ are orthonormal. The
kinetic energy of the determinantal state $\psi$ has the bound
given in \cite{Lieb}
\begin{equation}\label{so91}
 T(\psi)\le \left(1+C_N\right)E(\rho).
\end{equation}

Since we will establish a similar bound later, let us briefly recall the \
argument for \eqref{so91}. We start with \eqref{so3023}, where the first sum
on the r.h.s. is
$$
\sum_{k=1}^{N}E(\rho_k)=E(\rho).
$$
For the second sum we have to compute $E(\bj_k; \rho_k)$ for these
functions $\phi_k(\x)$,  and we have that
$$
  \rho_k(\x)|\bnab\chi_k(\x)|^2=
\frac{(2\pi)^2\rho(\x)}{N^3}k^2g(x^3)^4
$$
with the definition
\begin{equation}\label{g2}
g(u)^2=\int_{-\infty}^\infty\int_{-\infty}^\infty \,\rho(s,t,u)\,\rd s\rd t=
\frac{N}{2\pi}\frac{\rd}{\rd u}f(u).
\end{equation}
Hence
\begin{align}\label{est0}
\sum_{k\in\cK_N}E(\bj_k; \rho_k) &= \int_{\R^3}\rho_k(\x)
|\bnab\chi_k(\x)|^2\,\rd \x  \\\nonumber 
&=\frac{(2\pi)^2}{N^3}\sum_{k\in\cK_N}k^2 \int\limits_{-\infty}^\infty
\int\limits_{-\infty}^\infty \int\limits_{-\infty}^\infty \rho(s,t,u) 
g(u)^4 \rd s \rd t \rd u \\\nonumber
&=\frac{(2\pi)^2}{N^3}\sum_{k\in\cK_N}
k^2\int\limits_{-\infty}^{\infty}g(u)^6\,\rd u.
\end{align}
As shown in \cite{Lieb}
\begin{equation}\label{liebest}
\int_{-\infty}^{\infty}g(u)^6 \,\rd u\,\le 4 N^2 E(\rho).
\end{equation}
For the convenience of the reader we briefly recall its proof.
Since 
$$
g(u)^2=2\int_{-\infty}^u g(v)
\frac{\rd g(v)}{\rd v}\, \rd v
$$
holds, we conclude by the Schwarz inequality that for all $u$
\begin{equation}\label{g4}
g(u)^4\le 4\int_{-\infty}^{\infty} g(v)^2\, \rd v\;
\int_{-\infty}^{\infty}\left(\frac{\rd g(v)}{\rd v}\right)^2\, \rd v \equiv
P 
\end{equation}
The first integral on the r.h.s equals $N$ by the normalization 
condition on $\rho$. Therefore we obtain the estimate
\begin{equation}\label{g40}
 \int_{-\infty}^{\infty}g(u)^6 \, \rd u\le P \int_{-\infty}^{\infty}g(u)^2
\rd u =  4N^2\int_{-\infty}^{\infty}\left(\frac{\rd g(v)}{\rd v}\right)^2\,
\rd v.
\end{equation}
To conclude the proof of \eqref{liebest} we must show that
\begin{equation}\label{g41}
 \int_{-\infty}^{\infty}\left(\frac{\rd g(v)}{\rd v}\right)^2\, \rd v\le
E(\rho)
\end{equation}
holds. To do this write  
\begin{equation}\label{g42}
 \int_{-\infty}^{\infty}\left(\frac{\rd g(v)}{\rd v}\right)^2\,\rd
v=\frac{1}{4} \int_{-\infty}^\infty
\frac{\left(\int_{-\infty}^\infty\int_{-\infty}^\infty\frac{\partial}{
\partial v}\rho(x,y,v)\rd x \rd y\right)^2}
{\int_{-\infty}^\infty\int_{-\infty}^\infty\rho(x^\prime,y^\prime,
v)\rd x^\prime \rd y^\prime}\rd v,
\end{equation}
and then use 
\begin{align}\label{g43}
& \left(\int_{-\infty}^\infty\int_{-\infty}^\infty\frac{\partial}{\partial
v}\rho(x,y,v)\rd x \rd y\right)^2=
\left(\int_{-\infty}^\infty\int_{-\infty}^\infty
2\sqrt{\rho(x,y,v)}\frac{\partial}{\partial v}\sqrt{\rho(x,y,v)}\rd x
\rd y\right)^2
\\\nonumber
&\le 4\int_{-\infty}^\infty\int_{-\infty}^\infty
\rho(x,y,v)\rd x\rd y\int_{-\infty}^\infty\int_{-\infty}^\infty 
\left(\frac{\partial}{\partial v}\sqrt{\rho(x,y,v)}\right)^2 \rd x \rd y ,
\end{align}
by Schwarz's inequality. Insert this bound into the r.h.s. of \eqref{g42}
thereby  proving  \eqref{g41}. 
We insert inequality \eqref{liebest} into \eqref{est0} and 
perform the sum over $k$. Collecting terms yields \eqref{so91}.
\begin{theorem}\label{4:prop:1}
Assume $\bv(\x)=\bnab \tau(\x)$ for some function $\tau$, i.e., $\bv$ is curl free.
For given $N\ge 1$ let 
\begin{equation}\label{so16}
\chi_k(\x)=\tau(\x)+k f(x^3), \quad k\in\cK_N.
\end{equation}  
The revised functions 
\begin{equation}\label{so17}
\phi_k(\x)=\sqrt{\frac{\rho(\x)}{N}}\exp\{\ii(\tau(\x)+k f(x^3))\},
\end{equation}
with $f$ given by \eqref{so8}, form an orthonormal system. The determinantal
state $\psi$ 
\begin{equation}\label{so18}
\psi(\x_1,\cdots,\x_N)=\frac{1}{N!^{1/2}N^{N/2}}
\prod_{k=1}^N\left(\sqrt{\rho(x_k)}\e^{\ii \tau(\x_k)}\right)
\cdot\det\{\e^{\ii k f(x^3_j)}\},
\end{equation}
satisfies $\rho_{\psi}(\x)=\rho(\x)$ and $\bj_{\psi}(\x)=\bj(\x)=\rho(\x)\bv(\x)$
with the energy bound
\begin{equation}\label{so180}
T(\psi)\le C_N E(\rho)+ E(\bj;\rho).
\end{equation}
\end{theorem}

\begin{proof}
Clearly, the relation $\bj_{\psi}=\bj=\rho\bv$ follows from the fact that
$\sum_{k\in\cK_N } k =0$.
The proof of the first part follows from \cite{Lieb}. 
So we only have to prove the estimate \eqref{so180}. 
By \eqref{so16} we have 
\begin{align}
\bnab\chi_k(\x)&=\bnab\tau(\x)
+k\bnab f(x^3)\\\nonumber
&=\bv(\x)
+\frac{2\pi}{N} k g(x^3)^2\,\mathbf{e}_3,
\end{align}
where $\mathbf{e}_3=(0,0,1)$ is the unit vector in the 3-direction.

Since $\sum_k k =0$, we have that 
$$
\sum_k |\bnab\chi_k(\x)|^2 =
N |\bv(\x)|^2+\frac{4\pi^2 }{N^2}  g(x^3)^4 \sum_k k^2,
$$
i.e., the cross term vanishes. 
Combined with $\rho_k=\rho/N$ (and the fact that $\rho_k$ is independent
of $k$) 
this gives the inequality
\begin{equation}\label{est}
\sum_k \int_{\R^3}\rho_k(\x)\,|\bnab\chi_k(\x)|^2
\, \rd \x\le 
\frac{1}{N}E(\bj;\rho)+\frac{4\pi^2 }{N^2}\sum_k
k^2\int_{-\infty}^{\infty}
g(x^3)^6 \,\rd x^3.
\end{equation}
Using \eqref{g2} and summing over $k$ gives the bound \eqref{so180}.
\end{proof}

In sumary: the curl freeness of $ \bv$ is a sufficient condition 
for solving the current-density problem.

\section{\bf Solution of the current-density problem for a non curl free velocity field 
when $\bf N\ge 4$}\label{sec:N>3}

This section is devoted to a proof of 
\begin{theorem}\label{6:theo:1}
Given $\rho$ and $\j$, when $N\ge 4$ 
there is always a determinantal wave function $\psi$ with $\rho_\psi =\rho$ and
$\bj_\psi=\bj$. Moreover, if in addition $E(\rho)<\infty,\; E(\j,\rho)<\infty$ and if the curl $\bw=\bnab\times \v$ 
of \, $\v = \j/\rho$\,  and its first order derivatives satisfy the bounds
\begin{align}\label{curlbound}
\sup_{\x\in\R^3,\,j=1,2,3}\left(1+(x^1)^2\right)^{(1+\delta)/2}\left(1+(x^2)^2\right)^{(1+\delta)/2}
\left(1+(x^3)^2\right)^{(1+\delta)/2}|w_j(\x)|&<\infty\\\nonumber
\sup_{\x\in\R^3,\, i,j=1,2,3}\left(1+(x^1)^2\right)^{(1+\delta)/2}\left(1+(x^2)^2\right)^{(1+\delta)/2}
\left(1+(x^3)^2\right)^{(1+\delta)/2}|\partial_iw_j(\x)|&<\infty
\end{align}
for some $\delta>0$,
then $T(\psi)< \infty$.
\end{theorem}

We conjecture that condition (\ref{curlbound}) can be considerably
loosened. 
We have used the notation $\partial_i=\partial/\partial_{x^i}$, 
The proof will be split into several steps. (To
avoid clutter we
will sometimes omit the dependence on $\x$ from now on,
when the meaning is
clear. Recall that $\x= (x^1,x^2,x^3)$ and do not confuse $x^2$ with
$|\x|^2$.)

\underline{\bf Step 1} {(\it Construct the $\rho_i$):} We do this  in such a
way that all $\rho_i$ for $i\ge 4$ are equal, while the $\rho_i$ for $1\le i\le 4$ are
different. The motivation for this is that in the case where the
velocity field 
is not curl free, we cannot choose all $\rho_i$ to be equal to $\rho/N$.
Indeed, such an Ansatz would give
$$
\bv(\x) =\frac{1}{\rho(\x)}\bj(\x)=\bnab\frac{1}{N}\sum_{i=1}^N\chi_i(\x)
$$
by \eqref{so301}, which shows that ${\rm curl}\,\bv=0$, and which is a
contradiction. However, we may and will choose $N-3$ of them to
be equal. 
Set
$$
\xi(x)=\frac{1}{m}\int_{-\infty}^x\frac{1}{(1+y^2)^{(1+\delta)/2}}dy
$$
with 
$$
m=\int_{-\infty}^\infty\frac{1}{(1+y^2)^{(1+\delta)/2}}dy
$$
$\xi(x)$ is a continuous, strictly increasing function in $x$ with $\xi(-\infty)=0$ and $\xi(\infty)=1$.
$\delta$ is the $\delta$ in \eqref{curlbound} if the curl $\bw$ of $\v$ satisfies the bound 
\eqref{curlbound} and is arbitrary $>0$ otherwise. 
Set 
$\rho_i  = \eta_i \rho$ with
\begin{align}\label{cons1}
\eta_1(\x)&=\frac{2}{N}\xi(x^1+\alpha)\\\nonumber
\eta_2(\x)&=\frac{2}{N-1}\xi(x^1+\beta)(1-\eta_1(\x))\\\nonumber
\eta_3(\x)&=\frac{2}{N-2}\xi(x^2+\gamma)(1-\eta_1(\x)-\eta_2(\x))\\\nonumber
\eta_i(\x)&=\frac{1}{N-3}(1-\eta_1(\x)-\eta_2(\x)-\eta_3(\x)),\quad 4\le i\le N.
\end{align}
$\alpha,\beta,\gamma$ are real and, for the moment, arbitrary. 

Observe
that  $\eta_1$ and $\eta_2$ are functions of the first component
$x^1$ of $\x$ only, while the $\eta_j$ for $j\ge 3$ depend on $x^1$ and
$x^2$
but not on $x^3$. We claim $0\le 1-\eta_1(\x)-\eta_2(\x)$ and 
$0\le 1-\eta_1(\x)-\eta_2(\x)-\eta_3(\x)$ hold and thus $0\le \eta_j$ for all $1\le j\le N$. 
Indeed, an easy calculation gives
\begin{align*}
1-\eta_1(\x)-\eta_2(\x)&=(1-\eta_1(\x))\left(1-\frac{2}{N-1}\xi(x^1+\beta)\right)\\
1-\eta_1(\x)-\eta_2(\x)-\eta_3(\x)&=(1-\eta_1(\x))\left(1-\frac{2}{N-1}\xi(x^1+\beta)\right)
\left(1-\frac{2}{N-2}\xi(x^2+\gamma)\right)
\end{align*}
and this combined with 
\begin{align*}
 1-\frac{2}{N-1}\xi(x^1+\beta)\ge \frac{N-3}{N-1}\\
1-\frac{2}{N-2}\xi(x^2+\gamma)\ge \frac{N-4}{N-2}
\end{align*}
proves the claim.

As a consequence 
\begin{align}\label{cons4}
0&\le \rho_1\le \frac{2}{N}\rho,\quad
0\le \rho_2 \le\frac{2}{N-1}(\rho-\rho_1),\quad
0\le \rho_3\le \frac{2}{N-2}(\rho-\rho_1-\rho_2),\\\nonumber
0&\le\rho_i= 
\frac{1}{N-3}(\rho-\rho_1-\rho_2-\rho_3),\quad4\le i\le N, \quad\quad\quad
\sum_{i=1}^N\rho_i=\rho.
\end{align}

To fix $\alpha$, consider the function 
$$
I(\alpha)=\int_{\R^3}\rho_1(\x)\rd\x
=\frac{2}{N}\int_{\R^3}\xi(x^1+\alpha)\rho(\x)\rd\x,
$$
which is continuous and monotonically strictly increasing in
$\alpha$ (since $\xi$ has these properties). 
Since $\lim_{\alpha \to -\infty}I(\alpha) =0$ and
$\lim_{\alpha \to +\infty}I(\alpha)=2$, these properties imply that
there is a unique $\alpha$ such that $I(\alpha)=1$. 
We choose this value of $\alpha$ since it implies that 
$\int_{\R^3}\rho_1 d\x =1$, as required.
Having thus fixed $\alpha$, by the same argument and using the fact
that $\int_{\R^3}(\rho-\rho_1)d\x=N-1$, we can
fix $\beta$ uniquely such that also $\int_{\R^3}\rho_2=1$
is valid. Similarly, we can fix $\gamma$ such that also 
$\int_{\R^3}\rho_3=1$ is valid.
But then we also have that, for $4\le i\le N $,
\begin{align}\label{cons5}
\int_{\x\in\R^3}\rho_i(\x)\rd\x&=
\frac{1}{N-3}\int_{\R^3}(\rho(\x)-\rho_1(\x)-
\rho_2(\x)-\rho_3(\x))\rd\x =1.
\end{align}

This completes the construction of all one-body
densities 
$\rho_i$, $1\leq i \leq N$.

\medskip
\underline{\bf Step 2} {\it (Construction of phase functions
$\chi_i$ satisfying \eqref{so301})} We
postpone the
implementation of
the orthogonality to the remaining Steps 3-5.
Given the $\rho_i$ and $\eta_i$ constructed in the previous step, 
equation \eqref{so301} takes the equivalent form
\begin{equation}\label{chi01}
\sum_{i=1}^3\eta_i \bnab\chi_i
+\eta_4\bnab\left(\sum_{i=4}^N\chi_i\right)=\v.
\end{equation}
Recall that we assumed $\v=\j/\rho$ to be well defined though $\rho$ may have 
zero's or even vanish in a region. 
As already mentioned in the Introduction the best way to avoid such problems is to 
assume $\rho$ and $\v$ to be given rather than $\rho$ and $\j$. The current 
$\j$ is then {\it defined} to equal $\rho\v$.

We introduce  
\begin{equation}\label{chi1}
\tau(\x)=\frac{1}{N-3}\sum_{i=4}^N \chi_i(\x).
\end{equation}
Then, with the auxiliary quantities
\begin{equation}\label{chi2}
\widehat{\chi}_i=\chi_i-\tau, \qquad i=1,2,3 \ ,
\end{equation}
equation \eqref{chi01} is equivalent to
\begin{equation}\label{cons7}
\bnab\tau=
\v-\sum_{i=1}^3\eta_i \bnab\,\widehat{\chi}_i,
\end{equation}
which in particular says that the r.h.s. has to be curl free.
The strategy for determining the phase factors is as follows. We will first 
determine the necessary form of the $\widehat{\chi}_k$ 
that makes the right hand side of \eqref{cons7} curl free. Equation 
\eqref{cons7} then defines $\tau$ up to an uninteresting additive
constant.
In Step 3 the $\chi_i$ for $4\le i\le N$ will be determined in such a
way 
that they
satisfy \eqref{chi1} and such that the resulting wave functions 
$\phi_i=\rho_i^{1/2}\exp \ii \chi_i, 4\le i\le N$ are orthogonal. 
For this we will follow the 
strategy used in the proof of Theorem \ref{4:prop:1}. Finally in Step 4
we will 
determine the $\widehat{\chi}_i$ and hence the $\chi_i, 1\le i\le 3$ via
\eqref{chi2} such that all $\phi_i, \  1\le i\le N$  are pairwise
orthogonal.

To implement these steps, we first take the curl of \eqref{cons7} and 
obtain the curl-freeness condition:
\begin{equation}\label{cons8}
\sum_{i=1}^3\bnab\,\eta_i\times 
\bnab\,\widehat{\chi}_i=\bw ={\rm curl}\  \bv,
\end{equation}
Using \eqref{cons1} we can write out
\eqref{cons8} in components:
\begin{align}\label{cons10}
\partial_2\eta_3\; \partial_3\widehat{\chi}_3&=w_1\\\nonumber
-\partial_1\eta_1\;\partial_3\widehat{\chi}_1
-\partial_1\eta_2\;\partial_3\widehat{\chi}_2
-\partial_1\eta_3\;\partial_3\widehat{\chi}_3&=w_2\\\nonumber
\partial_1\eta_1\;\partial_2\widehat{\chi}_1
+\partial_1\eta_2\;\partial_2\widehat{\chi}_2
+\partial_1\eta_3\;\partial_2\widehat{\chi}_3
-\partial_2\eta_3\;\partial_1\widehat{\chi}_3&=w_3.
\end{align}
Recall that $\eta_1$ and $\eta_{2}$ depend on $x^1$ only, while
$\eta_3$ depends on $x^1$ and $x^2$. As a consequence no partial
derivatives of the form $\partial_1\widehat{\chi}_1$ or 
$\partial_1\widehat{\chi}_2$ appear in 
these equations.
As  preparation for the next step  we calculate some of the partial
derivatives of the $\eta$'s. 
The inequalities 
\begin{align}\label{cons9}
\partial_1 \eta_1(\x)&=\frac{2}{mN(1+(x^1+\alpha)^2)^{(1+\delta)/2}}>0\\\nonumber
\partial_2 \eta_3(\x)&
=\frac{2}{m(N-2)(1+(x^2+\gamma)^2)^{(1+\delta)/2}}(1-\eta_1(\x)-\eta_2(\x))>0
\end{align}
are valid due to
\begin{equation}\label{cons91}
 1-\eta_1(\x)-\eta_2(\x)>\frac{1}{6},
\end{equation}
an easy consequence of the definitions \eqref{cons1} of $\eta_1$ and $\eta_2$. 
In particular $\partial_1\eta_1$ and $\partial_2 \eta_3$ never vanish. 

Let $h_1,h_2,h_3$ be arbitrary functions of $x^1$ only.
Define 
\begin{align}\label{phi1}
 \kappa_1(\x)&=\sum_{j=1}^3\kappa_{1,j}(\x)\\\nonumber
\kappa_2(\x)&=0\\\nonumber
\kappa_3(\x)&=\int_0^{x^3}\left(\frac{w_1}{\partial_2\eta_3}\right)(x^1,x^2,
s)\rd s
\end{align}
with 
\begin{align}\label{phi10}
 \kappa_{1,1}(\x)&=\;\;\; \frac{1}{\partial_1\eta_1(x^1)}
\int_0^{x^2}w_3(x^1,s,x^3=0)\rd s\\\nonumber&\qquad\qquad 
-\frac{1}{\partial_1\eta_1(x^1)}\int_0^{x^3}\left(w_2+\frac{
\partial_1\eta_3}
{\partial_2\eta_3}\,w_1\right)(x^1,x^2,t)\rd t\\\nonumber
\kappa_{1,2}(\x)&=\;\;\;\frac{\partial_1 h_3(x^1)}{\partial_1\eta_1(x^1)}\eta_3(x^1,x^2)
\\\nonumber
\kappa_{1,3}(\x)&=-\frac{1}{\partial_1\eta_1(x^1)}
((\partial_1\eta_2)\;h_2)(x^1). 
\end{align}
In terms of these quantities the functions $\widehat{\chi}_i$ are defined as
\begin{equation}\label{phi2}
 \widehat{\chi}_i=\kappa_i+h_i,\qquad i=1,2,3.
\end{equation}
Define 
\begin{equation}\label{phi3}
 \u=\v-\sum_{i=1}^3\eta_i\bnab\; \widehat{\chi}_i.
\end{equation}
We have the 
\begin{lemma}\label{lem:u}
$\u$ is curl free for arbitrary $h_1,h_2,h_3$.
\end{lemma}
By what has been said so far, it suffices to check 
that \eqref{cons10} is satisfied. We give the proof in Appendix \ref{app:1}. 
It is somewhat intricate and uses 
the fact that $\bw$ has zero divergence.
By this lemma $\u$ is a gradient field and we 
{\it define} $\tau$ to be the solution
to the equation $\bnab \tau=\u$. $\tau$ is unique up to a constant
and is therefore fixed uniquely by requiring it to vanish at the origin.

To sum up: We have determined $\widehat{\chi}_1,\widehat{\chi}_2,\widehat{\chi}_3$ 
and $\tau$ such that \eqref{cons7} holds.
Finally we set 
\begin{equation}\label{phi4}
\chi_i= \widehat{\chi}_i+\tau,\qquad i=1,2,3.
\end{equation}
Observe that by \eqref{phi1} and \eqref{phi10} all components of the curl $\bw$ of $\bv$ enter the definition 
of these three phase functions.

\medskip
\underline{\bf Step 3} {\it (Orthogonality for $4\leq i \leq N$).}
We construct suitable phase functions  $\widehat{\chi}_i,\,4\le i\le N$ 
to achieve  
the orthogonality of the corresponding $N-3$ one-body wave functions
$\widehat{\phi}_i$
\begin{equation}\label{step36}
\widehat{\phi}_i(x)=\rho_i(x)^{1/2}\e^{\ii\,\widehat{\chi}_i(x)},\qquad 4\le i\le N.
\end{equation}
To achieve this we refer to our discussion 
in Section \ref{sec:curlfree}. Set 
\begin{equation}\label{step31}
\widehat{\rho}=(\rho-\rho_1-\rho_2-\rho_3)=(N-3)\rho_4
\end{equation}
such that $\int\widehat{\rho}(\x)\;d\x=(N-3)$,
which puts us in a $(N-3)$-body context by which we may invoke the
discussion 
of Section 
\ref{sec:curlfree}. 
Indeed, the associated $N-3$-body current is 
\begin{equation}\label{jhat}
\widehat{\j}
=\sum_{i=4}^N\rho_i\bnab\,\chi_i=
\frac{\widehat{\rho}}{N-3}\sum_{i=4}^N\bnab\,\chi_i=
\widehat{\rho}\,\bnab\,\tau,
\end{equation}
and so the associated velocity field $1/\widehat{\rho}\;\;\widehat{\j}$ is a gradient field
equal to $\u$ by the construction of $\tau$, see the end of Step 3.

With 
\begin{equation}\label{step33}
\widehat{f}(x^3)=\frac{2\pi}{N-3}\int_{-\infty}^{+\infty}\int_{-\infty}^{+\infty}
\int_{-\infty}^{x^3}\;\widehat{\rho}(s,t,u)\,\rd s \rd t  \rd
u
\end{equation}
and for $4\le i\le N$ we adjust the notation in \eqref{so16}-\eqref{so17} 
to the present situation and set 
\begin{equation}\label{step34}
\widehat{\chi}_i(\x)=\left(i-4 -\frac{N-4}{2}\right)\widehat{f}(x^3), \quad 4\le i\le N.
\end{equation}
Observe that when $i$ runs through $4,5,\cdots N$, then $i-4 -(N-4)/2$ runs through the set $\cK_{N-3}$, see \eqref{so15}.
By the arguments given in the proof of Proposition \ref{4:prop:1}, the functions \eqref{step36}
form an orthonormal system of $N-3$ vectors. 

\underline{\bf Step 4} {\it (Orthogonality for $i=1,2,3$).}
Here we extend the orthonormal system \eqref{step36}
with the help of suitably chosen phases 
$\widehat{\chi}_1,\widehat{\chi}_2,\widehat{\chi}_3$ 
and wave functions 
\begin{equation*}\label{step361}
\widehat{\phi}_k(\x)=\rho_k^{1/2}(\x)\e^{\ii\,\widehat{\chi}_k(\x)},\qquad 1\le k\le 3,
\end{equation*}
to an orthonormal system $\widehat{\phi}_i$ of
$N$ vectors. This in turn means we have to find
suitable functions $h_1,h_2$ and $h_3$ as introduced in Step 2.
With this {\it Ansatz} 
the scalar products, which we have to make vanish, can be written as
\begin{align}\label{step41}
\langle \widehat{\phi}_k,\widehat{\phi}_i\rangle&=\int_{\R^3}
\e^{-\ii (\kappa_k+h_k-\widehat{\chi}_i)}\,
\eta_k^{1/2}\,\eta_i^{1/2}\, \rho\, \rd\x \\\nonumber
\langle \widehat{\phi}_k,\widehat{\phi}_l\rangle&=\int_{\R^3}
\e^{-\ii(\kappa_k+h_k-\kappa_l-h_l)}\,
\eta_k^{1/2}\, \eta_l^{1/2}\, \rho\, \rd\x 
\end{align}
for $1\le l,k \le 3,l<k, 4\le i\le N$.
We invoke the following theorem in \cite{Lazarev}.
\begin{theorem}\label{6:lem:1}
Let $m\ge 1$ functions $\psi_j\in L^1(\R^n), 1\le j\le m$ 
be given. Then there exists a real, infinitely differentiable 
function $\chi(\x)$ on $\R^n$, with bounded derivatives, such that 
\begin{equation}\label{step410}
\int_{\R^n} \e^{-\ii\chi(\x)}\psi_j(\x)\rd\x=0
\end{equation}
holds for all $1\le j\le m$.
\end{theorem}
The $\chi(\x)$ constructed in \cite{Lazarev} is a function of one 
variable only (which may be taken to be any one of the $x^i$ that one
wishes)  and vanishes outside a bounded set in that variable. Consequently,
$\chi(\x)$ has bounded derivatives.
This implies that if the $\psi_j$ have finite kinetic energy (i.e., $\bnab
\psi \in L^2(\R^n)$) then
the functions $e^{-\ii\chi(\x)}\psi_j(\x)$ also have finite kinetic energy.
Unfortunately, the theorem in \cite{Lazarev} or the one in
\cite{Rutherfoord} does not tell us how large the kinetic energies of the 
$e^{-\ii\chi(\x)}\psi_j(\x)$ functions are, only that they are finite. 

Theorem \ref{6:lem:1} is a generalization of the Hobby-Rice theorem
\cite{Hobby-Rice}, see also \cite{Pincus},
according to which a piecewise constant $\chi(x)$ (equal to $0$ or $\pi$
everywhere) exists with the property stated in Theorem \ref{6:lem:1}. Such a
$\chi$ would necessarily lead to infinite kinetic energy (because of the
discontinuities) 
and would not be suitable for us. Theorem \ref{6:lem:1} tells us how to smooth out
the discontinuities, and is essential for us.

Theorem \ref{6:lem:1} can be used to orthogonalize any set of any $N$ functions, 
$f_1, \cdots, f_N$. 
It says that one can add a phase to $f_2 $ so that $f_1 $ and $f_2$ are orthogonal.
Then one can add a phase to $f_3$ so that $f_3$ is orthogonal to $f_1$ and $f_2$. 
Finally, one can make $f_N$ orthogonal to $f_1, \cdots, f_{N-1}$.

In our case we have to proceed cautiously. We will use the three undetermined functions
$h_1,\, h_2,  \,h_3$ as phases, but the astute reader will notice that our functions already
depend explicitly on $h_3$ and $h_2$ and might complain about lack of independence. In fact, 
only  $\psi_1$ depends on $h_3$ and $h_2$. Thus, no problem arises if we do things in 
the right order: First we determine $h_3$ to make $\psi_3$ orthogonal to $\psi_i$ for 
$i\geq 4$. This fixes $h_3 $. Then we fix $h_2$ similarly. Now $\psi_1$ is fixed and we 
are free to choose $h_1$ to complete the orthogonalization. The order is important!

We first consider the case $k=3$ in the first relation in \eqref{step41}. The aim is 
to find a suitable function $h_3$ depending on $x^1$ 
only. When we set 
$$
\psi_i^{(3)}(x^1)=\int_{(x^2,x^3)\in\R^2}
\e^{-\ii(\kappa_3(\x)-\widehat{\chi}_i(\x))}\eta_3
^{1/2}(\x)\,\eta_i^{1/2}(\x)\,\rho(\x)\rd x^2\rd x^3, 
\qquad 4\le i\le N,
$$
an element of $L^1(\R)$, we obtain
$$
\langle \widehat{\phi}_3,\widehat{\phi}_i\rangle
=\int_{x^1\in\R}\e^{-\ii h_3(x^1)}\psi_i^{(3)}(x^1)\rd x^1,\qquad 4\le i\le
N.
$$
By the previous lemma we can find a continuously differentiable function $h_3$ such 
that all these expressions vanish. This choice of $h_3$ determines $\widehat{\chi}_3$.
We turn to the case $k=2$ and introduce the following 
functions in $L^1(\R)$
\begin{align*}
\psi_3^{(2)}(x^1)&=\int_{(x^2,x^3)\in\R^2}
\e^{-\ii(\kappa_2(\x)-\kappa_3(\x)-h_3(x^1))}
\eta_2^{1/2}(\x)\eta_3^{1/2}(\x)\,\rho(\x)\rd x^2 \rd x^3 \\
\psi_i^{(2)}(x^1)&=\int_{(x^2,x^3)\in\R^2}
\e^{-\ii(\kappa_2(\x)-\widehat{\chi}_i(\x))}\eta_2^{1/2}(\x)\eta_i^{1/2}(\x)
\,\rho(\x)\rd x^2\rd x^3,
\qquad \qquad 4\le i\le N,
\end{align*}
such that
\begin{equation*}
\langle \widehat{\phi}_2,\widehat{\phi}_i\rangle
=\int_{\R}\e^{-\ii h_2(x^1)}\psi_i^{(2)}(x^1)\rd x^1,
\qquad\qquad\qquad\qquad 3\le i\le N.
\end{equation*}
Again by the lemma there is a continuously differentiable function $h_2$ in the 
variable $x^1$ such that all these expressions vanish. This choice of $h_2$ determines $\widehat{\chi}_2$.
Finally we turn to the case $k=1$. 
Set 
\begin{align*}
\psi_2^{(1)}(x^1)&=\int_{(x^2,x^3)\in\R^2}
\e^{-\ii(\kappa_1(\x)-\widehat{\chi}_2(\x)}
\eta_1(\x)^{1/2}\eta_2(\x)^{1/2} \,\rho(\x)dx^2dx^3, \\
\psi_3^{(1)}(x^1)&=\int_{(x^2,x^3)\in\R^2}
\e^{-\ii(\kappa_1(\x)-\widehat{\chi}_3(\x)}
\eta_1(\x)^{1/2}\eta_3(\x)^{1/2}
\,\rho(\x)\rd x^2\rd x^3, \\
\psi_i^{(1)}(x^1)&=\int_{(x^2,x^3)\in\R^2}
\e^{-\ii(\kappa_1(\x)-\widehat{\chi}_i(x))}\eta_1^{1/2}(\x)\,\eta_i^{1/2}(\x)\,
\rho(\x)\rd x^2\rd x^3, 
\quad 4\le i\le N,
\end{align*}
which again are elements of $L^1(\R)$.
Observe that $\kappa_1$ is known since $h_3$ and $h_2$ have been determined, see \eqref{phi1} and \eqref{phi10}.
By construction
\begin{equation*}
\langle \widehat{\phi}_1,\widehat{\phi}_i\rangle
=\int_{\R}\e^{-\ii h_1(x^1)}\psi_i^{(1)}(x^1)\rd x^1,
\qquad 2\le i\le N,
\end{equation*}
holds. We use the lemma a final time to find a function $h_1$ such that all these 
expressions vanish. 

To sum up, the $\widehat{\phi}_i$ for all $1\le i\le N$ form 
an orthonormal system. Then
$$
\phi_i(\x)=\widehat{\phi}_i(\x)\e^{\ii \tau(\x)}=\rho_i^{1/2}(\x)\e^{\ii\chi_i(\x)},
\qquad 1\le i\le N.
$$
are also orthonormal.  By construction
\eqref{so301} holds, and the proof of the first part of 
Theorem \ref{6:theo:1} is finished. 
It remains to prove $T(\psi)<\infty$ when \eqref{curlbound} holds. For this we take recourse to \eqref{so3023}.
Since
$$
\bnab\,\sqrt{\rho_i}=\bnab\,\sqrt{\eta_i\rho}=(\bnab\,\sqrt{\eta_i})\sqrt{\rho}+
\eta_i\bnab\,\sqrt{\rho}
$$
holds, by using the definition \eqref{cons1} one easily checks that 
each $\bnab\,\sqrt{\eta_i}$ is bounded. Since $\sqrt{\rho}, |\bnab\,\sqrt{\rho}|\in L^2(\R^3)$
we conclude $|\bnab\,\sqrt{\rho_i}|\in L^2(\R^3)$, that is $E(\rho_i)<\infty$ for all $i$.
To estimate $E(\j_i,\rho_i)$ we proceed as follows. Since $\chi_i=\tau+\widehat{\chi}_i$ and 
$0\le \rho_i\le \rho$
\begin{equation}\label{tpsibound}
 \int_{\R^3} \rho_i |\bnab\,\chi_i|^2\,\rd\x\le 2\int_{\R^3} \rho |\bnab\,\tau|^2\rd\x
+2\int_{\R^3} \rho|\bnab\,\widehat{\chi}_i|^2\rd\x.
\end{equation}
First we consider the case $i\ge 4$. 
Then the second term on the r.h.s. is finite by choice of $\widehat{\chi}_i$ and 
the discussion in Section \ref{sec:curlfree}. By the definition of $\tau$ and relation 
\eqref{phi3}
\begin{equation}
 |\bnab\,\tau|^2\le 8|\v|^2+8\sum_{i=1}^3|\bnab\, \widehat{\chi}_i|^2.
\end{equation}
Since 
$$
\int_{\R^3} \rho|\v|^2\rd\x=E(\j,\rho)<\infty 
$$
by assumption, we are done if we can show that
\begin{equation}\label{tpsibound1}
\int_{\R^3} \rho |\bnab\, \widehat{\chi}_i|^2\rd\x <\infty
\end{equation}
holds for all $i=1,2,3$. Then incidentally the r.h.s. of \eqref{tpsibound} is finite for 
all $1\le i\le N$. By \eqref{phi2}
\begin{equation}\label{tpsibound2}
\int_{\R^3} \rho |\bnab\, \widehat{\chi}_i|^2\rd\x\le 2\int_{\R^3} \rho |\bnab\, \kappa_i|^2\rd\x+
2\int_{\R^3} \rho |\bnab\,h_i|^2\rd\x,\qquad i=1,2,3.
\end{equation}
The second term on the r.h.s. is finite by the choice of the $h_i$, Theorem \ref{6:lem:1} and the comment thereafter.
As for the first term, the case $i=2$ is trivial since $\kappa_2=0$. 
In the appendix \ref{app:2} we will prove 
\begin{lemma}\label{lem:kappai}
 The functions $|\bnab\, \kappa_i|$ for $i=1,3$ are bounded.
\end{lemma}
Given this lemma the first integral on the r.h.s. of \eqref{tpsibound2} is also finite thus completing the 
proof of theorem \ref{6:theo:1}.

\hfill \qedsymbol

\section{\bf The case $N=2,\bnab\times\bv\neq 0$}
\label{sec:N=2}
In this section we discuss the case of two particles, $N=2$. Surprisingly, 
we have not been able to provide conditions that are both necessary and
sufficient for a solution of the problem to exist. Of course curl freeness
of the velocity field is sufficient but not 
necessary as the first example shows. Conversely the second example provides
a (non curl free) velocity field, for which there is no solution.

\vspace{0.3cm}
\subsection{Solution to an example with
$N=2,$  {\bf and} $\bnab \times \bv\neq 0$.}~\\
Let $\bc \neq 0$ be a fixed vector and consider 
\begin{align}\label{so19}
\rho(\x)&=2 \pi^{-3/2}\:\e^{-|\x|^2}\\\nonumber
{\bj}(\x)&=\pi^{3/2}({\bc}\times {\x})\:\e^{-|\x|^2}
\end{align}
with resulting velocity field 
\begin{equation}\label{so20}
 {\bv}(\x)=\tfrac{1}{2}({\bc}\times {\x})
\end{equation}
which is not curl free.
The normalization $\int_{\R^3}\rho(\x)\rd \x =2$ holds, $\bnab \cdot \bj
(\x) =0$, 
and both $E(\rho)$ and
$E(\bj;\rho) $ are finite. 
We will consider the case where $\bc=(0,0,1)$, a general $\bc$ may be discussed similarly.

Assume there are 
$\rho_1> 0 ,\rho_2>0,\chi_1$ and $\chi_2$, with $\rho_1+\rho_2=\rho $, which
are solutions to the equation 
\begin{equation}\label{so21}
\rho_1 \bnab\, \chi_1+
\rho_2\bnab\, \chi_2=\bj
\end{equation}\
subject to the condition
\begin{equation}\label{so211}
\int_{\R^3}\rho_j(\x)^{1/2}\rho_k(\x)^{1/2}\e^{\ii(\chi_k(\x)-\chi_j(\x))}\rd \x
=\delta_{jk}.
\end{equation}
Introduce $\eta_k(\x)=\rho_k(\x)/\rho(\x)$ 
which satisfy $0\le \eta_k \le 1,\;\eta_1+\eta_2=1$. So we may 
rewrite \eqref{so21} as 
\begin{equation}\label{so212}
\eta_1\bnab\, \chi_1
+\eta_2 {\bnab}\,\chi_2=\bv,
\end{equation}
or equivalently 
\begin{equation}\label{so2121}
 \eta_1\bnab\,(\chi_1- \chi_2)+\bnab \chi_2=\bv,
\end{equation}
and \eqref{so211} as
\begin{equation}\label{so2122}
\int_{\R^3}\eta_j(x)^{1/2}\eta_k(\x)^{1/2}\e^{\ii(\chi_k(\x)-\chi_\bj(\x))}\rho(\x)\rd \x=\delta_{jk}.
\end{equation}
We take the curl of \eqref{so212} and use $\bnab\,\eta_2(\x)=
-\bnab\, \eta_1(\x)$, a consequence of the
relation 
$\eta_1(\x)+\eta_2(\x)=1$. 
This gives
\begin{equation}\label{so214}
\bnab\, \eta_1(\x)\times 
\bnab\, \widehat{\chi}(\x)
=\bnab  \times\v(\x)={\bc}
\end{equation}
with $\widehat{\chi}=\chi_1-\chi_2$ and valid for all $\x\in\R^3$. 
As a consequence of \eqref{so214} the vector fields
$\bnab\, \eta_1(\x)$ and 
$\bnab\, \widehat{\chi}(\x)$ are never parallel 
and in particular never vanishing. 
In addition we conclude that they are orthogonal to ${\bc}$.

We define 
\begin{equation}\label{def1}
\eta_1(\x)=\tfrac{1}{2}(1+\tanh x^1), \quad \widehat{\chi}(\x)=
2x^2\cosh ^2x^1 +h(x^1)
\end{equation}
where for the moment $h$ is an arbitrary function of $x^1$ alone. 
In particular
\begin{equation}\label{def2}
\eta_2(\x)=\tfrac{1}{2}(1-\tanh x^1)
\end{equation}
and $0\le\eta_j\le 1$ is satisfied. 
Also by construction
$$\bnab\times(\eta_1
\bnab\, \widehat{\chi}-\v)=\bc,
$$
{\it i.e.} \eqref{so214} is satisfied. But this implies there is a solution 
$\chi_2$ to \eqref{so2121}. More explicitly
\begin{align}\label{def3}
\chi_2(\x)&= -x^2 \left((1+\tanh x^1)\cosh^2 x^1-x^1\right)\\\nonumber
&\qquad-\tfrac{1}{2}\int_0^{x^1}(1+\tanh y)\frac{d}{dy}h(y)\rd y+const.,
\end{align}
where  $h(y)$ is undetermined as yet. $\chi_1$ is of course given as
$\widehat{\chi}-\chi_2$. 
Moreover, since $\tanh$ is odd, 
\begin{equation}\label{def4}
\int_{\R^3}\eta_1\,\rho\, \rd \x=\int_{\R^3}\eta_2\,\rho\, \rd \x
=\tfrac{1}{2}\int_{\R^3}\rho\, \rd x=1.
\end{equation}
Thus \eqref{so211} is satisfied for $j=k=1,2$ for any choice of $h$.
To determine $h$, we inspect the remaining condition $(j=2,k=1)$ in \eqref{so211}, which 
we write in the form
\begin{align}\label{def5}
\int_{R^3}\eta_1\eta_2
\rho\,\e^{\ii\widehat{\chi}}\,\rd \x&=\frac{1}{2\pi^{3/2}}
\int _{\R^3}\sqrt{1-\tanh^2 x^1}\;
\e^{-(x^1)^2-(x^2)^2-(x^3)^2}\e^{\ii(2 x^2\cosh ^2 x^1 +h(x^1))}\rd \x\\\nonumber
&=0.
\end{align}
Set 
\begin{align*}
g(x^1)&=\frac{1}{2\pi^{3/2}}
\sqrt{1-\tanh^2 x^1}\;\e^{-(x^1)^2}\int_{(x^2,x^3)\in \R^2}\e^{-(x^2)^2-(x^3)^2}\,
\e^{\ii 2 x^2\cosh ^2 x^1 }\,\rd x^2 \rd x^3\\
&=\frac{1}{2\pi^{1/2}}\sqrt{1-\tanh^2 x^1}\;\e^{-(x^1)^2}\e^{-\cosh^4 x^1},
\end{align*}
which is integrable and positive for all $x^1$.
Condition \eqref{def5} takes the form
\begin{equation}\label{def6}
\int_{-\infty}^{\infty}g(x^1)\e^{\ii h(x^1)}\rd x^1=0.
\end{equation}
Set 
$$
a=\int_{-\infty}^{+\infty}g(y)\rd y>0.
$$
Then the choice
\begin{equation}\label{def7}
h(x^1)=\frac{2\pi}{a}\int_{-\infty}^{x^1} g(y) \rd y
\end{equation}
with $h(-\infty)=0,h(\infty)=2\pi$ gives
$$
\int_{-\infty}^{\infty}g(x^1)\e^{\ii h(x^1)}\rd x^1
=\frac{a}{2\pi\ii }\int_{-\infty}^{\infty}\frac{\rd}{\rd x^1}\e^{\ii
h(x^1)}\rd x^1
=\frac{a}{2\pi\ii }\left(\e^{\ii h(\infty)}-\e^{\ii h(-\infty)}\right)=0.
$$
By inserting this solution for $h$ into \eqref{def3}, all quantities are
determined. We claim it gives a solution $\psi$ for which $T(\psi)<\infty$. 
For the proof we use the identity \eqref{so3023}. First
$E(\rho_1)<\infty$ and $E(\rho_2)<\infty$ is 
an easy consequence of $E(\rho)<\infty$ and the choice of $\eta_1$ and $\eta_2$. 
An easy calulation gives the bound
$$
|\bnab\, \chi_2(\x)|\le 2|x^2\sinh2x^1|+2\cosh^2x^1+|x^1|+\left|\frac{\partial h(x^1)}{\partial x^1}\right|.
$$
But 
$$
\frac{\partial h(x^1)}{\partial x^1}=\frac{2\pi}{a}g(x^1)
$$
decreases strongly as $x^1\rightarrow \pm\infty$. Since $\rho_2<\rho$, and
thanks to the Gaussian form 
of $\rho$, we therefore obtain
$$
E(\j_2;\rho_2)=\int_{\R^3}\rho_2|\bnab\, \chi_2|^2\,\rd\x<\infty.
$$
As for $E(\j_1;\rho_1)$ we use $|\bnab\,\chi_1|\le |\bnab\chi_2|+|\bnab\,\widehat{\chi}|$ combined with the 
estimate
$$
|\bnab\,\widehat{\chi}|\le |x^2\sinh 2x^1|+\left|\frac{\partial h(x^1)}{\partial x^1}\right|+2\cosh^2 x^1,
$$
a consequence of the definition \eqref{def1} of $\widehat{\chi}$. So we may use the same arguments 
as for the proof of $E(\j_2;\rho_2)<\infty$ to conclude $E(\j_1;\rho_1)<\infty$. By \eqref{so3023} 
this proves the 
claim $T(\psi)<\infty$.


\vspace{0.3cm}
\subsection{No solution to an example with
$N=2,  \bnab\times {\bv}\neq 0$.}~\\
For the two-body case ($N=2$) we will provide an example
with $\bnab\times\v\neq 0$, 
for which there is no continuously differentiable solution to the problem. 
This example originated out
of discussions with Th. Br\"ocker \cite{Broecker0}.
Another, older example is by Taut, Machon and Eschrig \cite{Taut}. 
\begin{example}[N=2] \label{5:ex:2}
Consider the choice 
\begin{align}\label{counter1}
\rho(\x)&=\frac{2}{\pi^{3/2}}\:\e^{-\x^2}\\\nonumber
\j(\x)&=\frac{\pi^{3/2}}{2}(0,-2x^1\,x^3,-x^1\,x^2)\:\e^{-x^2}
\end{align}
with resulting velocity field 
\begin{equation}\label{counter2}
 \v(\x)=(0,-2x^1\,x^3,-x^1\,x^2).
\end{equation}
Clearly $E(\rho)<\infty, E(\j,\rho)<\infty$.
\end{example}
\begin{proposition}[N=2]\label{5:prop:2}
There exists no solution to the problem with continuously 
differentiable $\rho_k/\rho$ and $\chi_k,\: k=1,2$, when $\rho$ and 
$\j$ are of the form \eqref{counter1}.
\end{proposition}
\begin{remark}\label{5:re:1}
We have not been able to show that there is no solution 
$\rho_1,\j_1,\rho_2,\j_2$ to the problem, when the solution is only required 
to satisfy 
$T(\psi)=E(\rho_1)+E(\j_1,\rho_1,)+E(\rho_2)+E(\j_2,\rho_2)<\infty$  
and which means less smoothness for $\rho_{1(2)},\j_{1(2)}$ and $\v_{1(2)}$.
\end{remark}
\begin{proof} 
Introduce the harmonic function on $\R^3$
\begin{equation}\label{counter3}
h(x^1,x^2,x^3)=\tfrac{1}{2}\left((x^1)^2+(x^2)^2-2(x^3)^2\right).
\end{equation}
An easy calculation shows that the curl of $\v$ equals the
gradient of $h$,
\begin{equation}\label{counter4}
\bnab\times\v(\x)=
\bnab\,h(\x)=(x^1,x^2,-2x^3).
\end{equation}
With the notation and discussion in the previous subsection, in particular 
in connection with the first relation in \eqref{so214},
we have to look for solutions
$\eta_1$ and $\widehat{\chi}$ to the relation
\begin{equation}\label{counter5}
\bnab\,\eta_1
\times\bnab\,\widehat{\chi}
=\bnab\,h.
\end{equation}
But now we claim there are no solutions to \eqref{counter5}. 
Indeed, there is even a stronger result due to Th. Br\"ocker \cite{Broecker0}, 
which reads as follows.
\begin{lemma} \label{5:lem:1}
Given the function $h$ \eqref{counter3}, there are no continuous
vector fields 
$\ba$ and $\bb$ on $\R^3$ with 
\begin{equation}\label{counter6}
\ba\times \bb=\bnab\,h.
\end{equation}
\end{lemma}
\begin{proof}
Assume to the contrary that there are 
solutions $\ba$ and $\bb$ to \eqref{counter6}.
Note that the vector field $\bnab\,h(\x)$ is non-vanishing 
for $\x\neq 0$. Hence the vector fields $\ba$ and $\bb$ necessarily 
share the same property and in addition we must have 
\begin{equation}\label{counter7}
\ba(\x)\perp\bnab\,h(\x),\qquad \x\neq 0 
\end{equation}
(and similarly for $\bb(\x)$). Condition \eqref{counter7} written out for 
$\ba(\x)=(a^1(\x),a^2(\x),a^3(\x))$ is 
\begin{equation}\label{counter8}
a^1(\x) x^1+a^2(\x) x^2-2a^3(\x)x^3=0.
\end{equation}
Introduce the vector field 
\begin{equation}\label{counter9}
\u(\x)=(a^1(\x),a^2(\x),-2a^3(\x))
\end{equation}
which, by the discussion just made,  is non-vanishing for $\x\neq 0$. By
\eqref{counter8} 
it is orthogonal to the radius vector $(x^1,x^2,x^3)$ . 
Hence it is tangential to any sphere centered at the origin 
and non-vanishing everywhere there. 
But this contradicts the {\it Hairy Ball Theorem} of Brouwer \cite{Brouwer}. 
For modern proofs of this theorem see {\it e.g.} \cite{Dold}, IV, 4.4, 
\cite{Spanier}, 
Chap.4, Sec.7, Corr.11. A proof using simple analytic tools is given in 
\cite{Broecker},
 VI, 2.4.
\end{proof}  
The proposition is now a direct consequence of this lemma and the preceding 
discussion.
\end{proof}
There is an easier direct proof, that there are no continuously 
differentiable solutions $\ba(\x)$ and $\bb(\x)$ to \eqref{counter6},  which
uses
a slightly stronger condition.
Indeed, make a Taylor expansion and write $\ba(\x)=
\widehat\ba+A\x+o(|\ x|^2)$ and similarly 
$\bb(\x)=\widehat\bb+B\x+o(|\ x|^2)$, where $A$ and $B$ are $3\times 3$
matrices. Also let $T=\diag(1,1,-2)$, whence
$T\x=\bnab\,h(\x)$. But then the condition 
\eqref{counter6} first says $\ba\times \bb=0$ and 
$\widehat\ba\times B\x-\widehat\bb\times A\x=T\x$. The first condition says 
that $\widehat\ba$ and $\widehat\bb$ are parallel. 
Now the case $\widehat\ba=\widehat\bb=0$ can be excluded immediately and so
we may assume 
that at least one vector is non-vanishing, say $\widehat\ba\neq 0$, and
that 
$\widehat\bb=\lambda \widehat\ba$. But with $\y=T\x$ this leads to the
relation
$$
\widehat\ba\times(-\lambda A+ B)T^{-1}\y=\y,
$$
valid for all small $\y$ and hence, by linearity, for all $\y$. In 
particular this means that $\widehat\ba$ is orthogonal to all $\y$, 
$\widehat\ba\perp \y$, which is a contradiction.


\begin{appendix}

\section{Proof of Lemma \ref{lem:u}}\label{app:1}~~\\
With $\kappa_3$ as given in \eqref{phi1}, $\widehat{\chi}_3=\kappa_3+h_3$ obviously solves the first relation in 
\eqref{cons10}. Inserting this into
the two other equations in \eqref{cons10} gives the equations
\begin{align}\label{1:app:1}
\partial_3(\partial_1\eta_1\;\widehat{\chi}_1+\partial_1\eta_2\;\widehat{\chi}_2)&=
\widehat{w}_2\\\nonumber
\partial_2(\partial_1\eta_1\;\widehat{\chi}_1
+\partial_1\eta_2\;\widehat{\chi}_2)&=\widehat{w}_3
\end{align}
with 
\begin{align}\label{1:app:2}
\widehat{w}_2(\x)&=-w_2(\x)-\frac{\partial_1\eta_3(x^1,x^2)}{\partial_2\eta_3(x^1,x^2)}\,w_1(\x)\\\nonumber
\widehat{w}_3(\x)&=\:w_3(\x) -\partial_1\eta_3(x^1,x^2)
\int_0^{x^3}\left(\partial_2\frac{w_1}{\partial_2\eta_3}\right)(x^1,x^2,s)\rd s
\\\nonumber
&\qquad+\partial_2\eta_3(x^1,x^2)
\int_0^{x^3}
\left(\partial_1\frac{w_1}{\partial_2\eta_3}\right)(x^1,x^2,s)\rd s
+\partial_2\eta_3(x^1,x^2)\partial_1 h_3(x^1).
\end{align}
Using the fact that $\bw$ has vanishing  divergence by its very definition,
a short calculation 
shows that the following necessary and sufficient condition for solving \eqref{1:app:1}
\begin{equation}\label{1:app:3}
\partial_2\widehat{w}_2(\x)=\partial_3\widehat{w}_3(\x)
\end{equation}
is valid for any choice of $h_3(x^1)$. So since \eqref{1:app:3} holds for each $x^1$, 
$(\widehat{w}_3,\widehat{w}_2)$ is a two-dimensional gradient field. In other words there exists $\widehat{w}$ 
such that $(\widehat{w}_3,\widehat{w}_2)=(\partial_2\widehat{w},\partial_3\widehat{w})$ holds. $\widehat{w}$ 
can be obtained by integrating this vector field, for example from $(x^1,0,0)$ - with arbitrary initial value 
$\widehat{h}_1(x^1)$ - to $(x^1,x^2,0)$ and 
from there to $(x^1,x^2,x^3)$. 
Thus
\begin{equation}\label{1:app:4}
 \widehat{w}(\x)=\int_0^{x^2}\widehat{w}_3(x^1,s,x^3=0)\rd s
+\int_0^{x^3}\widehat{w}_2(x^1,x^2,t)\rd t+\widehat{h}_1(x^1).
\end{equation}
So with our choice \eqref{phi1},\eqref{phi10} and \eqref{phi2} 
for $\widehat{\chi}_1,\widehat{\chi}_2,\widehat{\chi}_3$
and the choice $\widehat{h}_1(x^1)=\partial_1\eta_1(x^1) h_1(x^1)$ the relation 
$$
\partial_1\eta_1\;\widehat{\chi}_1+\partial_1\eta_2\;\widehat{\chi}_2=\widehat{w}
$$
is satisfied. Observe that 
$$
\widehat{w}_3(x^1,x^2,0)=w_3(x^1,x^2,0)+\partial_2\eta_3(x^1,x^2)\partial_1 h_3(x^1)
$$
holds. Therefore also relation \eqref{cons10} is valid and the proof of Lemma \ref{lem:u} is complete.

\section{Proof of Lemma  \ref{lem:kappai}}\label{app:2}
\renewcommand{\theequation}{\mbox{\Alph{section}.\arabic{equation}}}
\setcounter{equation}{0} 
We start with estimates for $\bnab\, \kappa_3$. By \eqref{phi1}
\begin{align}\label{1:app}
\partial_1\,\kappa_3(\x)&=-\frac{\partial_1\partial_2\eta_3(x^1,x^2)}{(\partial_2\eta_3(x^1,x^2))^2}
\int_0^{x^3}w_1(x^1,x^2,s)\rd s\\\nonumber
&\quad\quad+\frac{1}{\partial_2\eta_3(x^1,x^2)}\int_0^{x^3}\partial_1 w_1(x^1,x^2,s)\rd s\\\nonumber
\partial_2\,\kappa_3(\x)&=-\frac{\partial_2^2\eta_3(x^1,x^2)}{(\partial_2\eta_3(x^1,x^2))^2}
\int_0^{x^3}w_1(x^1,x^2,s)\rd s\\\nonumber
&\quad\quad+\frac{1}{\partial_2\eta_3(x^1,x^2)}\int_0^{x^3}\partial_2 w_1(x^1,x^2,s)\rd s
\\\nonumber
\partial_3\,\kappa_3(\x)&= \frac{1}{\partial_2 \eta_3(x^1,x^2)}w_1(\x).
\end{align}
To see that all $|\partial_j\,\kappa_3(\x)|,\;j=1,2,3$ are bounded, we proceed as follows.
Let $W$ stand for any of the quantities $w_1,\partial_1w_1,\partial_2w_1$. By the assumption 
\eqref{curlbound} there exists a constant $0<C_1<\infty$ such that 
\begin{equation}\label{estW0}
|W(\x)|\le C_1\left(1+(x^1)^2\right)^{-(1+\delta)/2}\left(1+(x^2)^2\right)^{-(1+\delta)/2}
\left(1+(x^3)^2\right)^{-(1+\delta)/2}
\end{equation}
holds. Therefore there is another constant $0<C_2<\infty$, such that the bound
\begin{align}\label{estW}
\Big|\int_0^{x^3}W(x^1,x^2,s)ds\Big|&\le \int_{-\infty}^{\infty}|W(x^1,x^2,s)|ds\\\nonumber
&\le C_2 \left(1+(x^1)^2\right)^{-(1+\delta)/2}\left(1+(x^2)^2\right)^{-(1+\delta)/2}
\end{align}
is valid. With this preparation we start with an estimate for the first term contributing to
$\partial_1\,\kappa_3$, which we call $A_1$. Now by \eqref{cons91}
\begin{align}\label{estW1}
\Big|\frac{\partial_1\partial_2\eta_3(x^1,x^2)}{(\partial_2\eta_3(x^1,x^2))^2}\Big|
&=\frac{N-2}{N}\frac{
\left(1+(x^1+\alpha)^2\right)^{-(1+\delta)/2}\left(1+(x^2+\gamma)^2\right)^{(1+\delta)/2}}
{(1-\eta_1(\x)-\eta_2(\x))^2}\\\nonumber
&\le C_3\left(1+(x^2+\gamma)^2\right)^{(1+\delta)/2}
\end{align}
for yet another finite constant $C_3$. We have used the relation
$$
\partial_1\partial_2\eta_3(x^1,x^2)=-\frac{4}
{m^2N(N-2)\left(1+(x^1+\alpha)^2\right)^{(1+\delta)/2}\left(1+(x^2+\gamma)^2\right)^{(1+\delta)/2}}
$$
Combining this estimate with the estimate 
\eqref{estW} for the choice $W=w_1$ and with the estimate 
\begin{equation}\label{eb2}
\sup_{x^2\in\R}(1+(x^2+\gamma)^2)^{(1+\delta)/2}\left(1+(x^2)^2\right)^{-(1+\delta)/2}<\infty
\end{equation}
shows that $A_1$ is bounded. 

As for the second contribution to $\partial_1\,\kappa_3$, and which we call $A_2$,
we use the estimate \eqref{estW} for $W=\partial_1w_1$ combined with
\begin{equation}\label{eb1}
\Big|\frac{1}{\partial_2\eta_3(x^1,x^2)}\Big|\le C_4\left(1+(x^2+\gamma)^2\right)^{(1+\delta)/2},
\end{equation}
which follows from \eqref{cons9} and \eqref{cons91}, and the estimate \eqref{eb2}
to conclude that $A_2$ is also bounded.

We turn to an estimate for $\partial_2\,\kappa_3$ and start with the first contribution, 
which we call $B_1$. The relation
$$
\partial_2 ^2\eta_3(\x)
=-\frac{2(1+\delta)(x^2+\gamma)}{m(N-2)\left(1+(x^2+\gamma)^2\right)^{(3+\delta)/2}}
(1-\eta_1(\x)-\eta_2(\x))
$$
gives the estimate
\begin{align}\label{eb3}
\Big|\frac{\partial_2^2\eta_3(x^1,x^2)}{(\partial_2\eta_3(x^1,x^2))^2} \Big|&=
\frac{(1+\delta)m(N-2)}{2}
\frac{|x^2+\gamma|\left(1+(x^2+\gamma)^2\right)^{-(1-\delta)/2}}{(1-\eta_1(\x)-\eta_2(\x))}\\\nonumber
&\le C_5\left(1+(x^2+\gamma)^2\right)^{\delta/2}\\\nonumber
&\le C_5\left(1+(x^2+\gamma)^2\right)^{(1+\delta)/2}.
\end{align}
Again we have used \eqref{cons91} and the trivial bound 
$$
|x^2+\gamma|\left(1+(x^2+\gamma)^2\right)^{-1}\le 1.
$$
We combine this bound with the bound \eqref{estW} for the choice $W=w_1$ and the bound \eqref{eb2}
to conclude that $B_1$ is bounded.

As for the second contribution to $\partial_2\,\kappa_3$ and which we call $B_2$, we proceed 
in analogy to 
the proof of the estimate of $A_2$. That is we use \eqref{eb1} and \eqref{eb2} and \eqref{estW} 
for the choice $W=\partial_2w_1$ to conclude that $B_2$ is bounded.

Finally we use \eqref{estW0} for the choice $W-w_1$ and \eqref{eb1} to conclude that 
$|\partial_3\,\kappa_3|$ is bounded.

As for $\kappa_1$, we start with 
$$
|\kappa_1(\x)|\le \sum_{j=1}^3|\kappa_{1,j}(\x)|.
$$
The boundedness of $|\bnab\,\kappa_{1,1}|$, see \eqref{phi10}, follows similar to one for $|\bnab \kappa_3|$. 
Due to the presence of the factor $\partial_1 h_3(x^1)\;\kappa_{1,2}(\x)$ is smooth, vanishes for all large 
$x^1$ and has bounded derivatives, that is $|\bnab\,\kappa_{1,2}|$ is bounded. By definition of $\kappa_{1,3}$ 
it remains to estimate 
\begin{align}\label{2:app}
\partial_1 \kappa_{1,3}(\x)&=-\left(\partial_1\frac{1}{\partial_1\eta_1}(\x)\right)
((\partial_1\eta_2)\;h_2)(\x)-\frac{1}{\partial_1\eta_1(\x)}
((\partial_1^2\eta_2)\;h_2)(\x)\\\nonumber&\qquad-\frac{1}{\partial_1\eta_1(\x)}
((\partial_1\eta_2)\;\partial_1h_2)(\x),
\end{align}
since $\kappa_{1,3}$ is a function of $x^1$ only.
The relations
\begin{align}\label{3:app}
\partial_1\frac{1}{\partial_1\eta_1}(\x)&=mN(1+\delta)x^1\left(1+(x^1)^2\right)^{(\delta-1)/2}\\\nonumber
\partial_1\eta_2(\x)&=\frac{1}{m(N-1)}\left(1+(x^1+\beta)^2)\right)^{-(1+\delta)/2}(1-\eta_1(\x))\\\nonumber
&\qquad-\frac{1}{mN(N-1)}\xi(x^1+\beta)\left(1+(x^1+\beta)^2)\right)^{-(1+\delta)/2}
\end{align}
show that
$$
\left(\partial_1\frac{1}{\partial_1\eta_1}\right)\partial_1\eta_2
$$
is bounded. Since $h_2$ is bounded this shows that the first term on the r.h.s of \eqref{2:app} is bounded. 
By calculating $\partial_1^2\eta_2$, a similar argument shows that 
$$
\frac{1}{\partial_1\eta_1}\partial_1^2\eta_2 
$$
is bounded, such that the second term on the r.h.s of \eqref{2:app} is also bounded.
The third term is bounded, since $\partial_1h_2(x^1)$ vanishes for all large $x^1$.
In conclusion, we have established that $|\bnab\,\kappa_{1,3}|=|\partial_1\,\kappa_{1,3}|$ is bounded and this 
completes the proof of Lemma \ref{lem:kappai}.
\end{appendix}


\end{document}